\documentclass[final,journal,letterpaper,10pt]{IEEEtran}
\usepackage{graphicx}
\usepackage{amsfonts}
\usepackage{times}
\usepackage{latexsym}
\usepackage{amssymb}
\usepackage{amsmath}
\usepackage{cite}
\usepackage{verbatim}
\usepackage{subfigure}
\usepackage{multirow}
\usepackage{lastpage}
\usepackage{calc}
\usepackage{eso-pic}
\usepackage{tabularx}
\usepackage{ifthen}
\usepackage{epstopdf}

\newtheorem{theorem}{Theorem}

\newtheorem{assumption}[theorem]{Assumption}

\newtheorem{corollary}[theorem]{Corollary}

\newenvironment{proof}{ \textbf{Proof:} }{ \hfill $\Box$}

\def\bb0{{\mathbb{0}}}

\def\bb{{\mathbf{b}}}

\def\b0{{\mathbf{0}}}

\def\bbE{{\mathbb{E}}}

\def\bbP{{\mathbb{P}}}

\def\bbR{{\mathbb{R}}}

\def\sf0{{\mathsf{0}}}

\begin{document}

\title{Asymptotic Coverage Probability and Rate in Massive MIMO Networks}

\author{
Tianyang Bai and Robert W. Heath, Jr.

\thanks{The authors are with the Wireless Networking and Communications
Group, the University of Texas at Austin, Austin, TX, USA.
(email: rheath@utexas.edu). This material is based upon work supported in part by the National Science Foundation under Grant No. NSF-CCF-1218338. }}

\maketitle
\begin{abstract}
Massive multiple-input multiple-output (MIMO) is a transmission technique for cellular systems that uses many antennas to support not-as-many users. Thus far, the performance of massive MIMO has only been examined in finite cellular networks. In this letter, we analyze its performance in random cellular networks with Poisson distributed base station locations. Specifically, we provide analytical expressions for the asymptotic coverage probability and rate in both downlink and uplink when each base station has a large number of antennas. The results show that, though limited by pilot contamination, massive MIMO can provide significantly higher asymptotic data rate per user than the single-antenna network.
\end{abstract}

\section{Introduction}
Massive multiple-input multiple-out (MIMO) communication is a multiuser transmission strategy where an order of magnitude more antennas than in conventional systems are used to serve a large number of users \cite{Marzetta2010}. Massive MIMO exploits pilot reuse to reduce training overhead and reciprocity in the channel to avoid feedback.

Prior work showed that high throughput could be achieved with a large number of antennas through simple matched filtering, eliminating the effects of noise, and that pilot contamination was a performance limiting factor \cite{Marzetta2010}. Asymptotic performance was characterized for a more comprehensive set of channels and beamforming designs in \cite{Hoydis2013}, but was still limited by pilot contamination. The impact of pilot contamination on the asymptotic SIR distribution was studied in \cite{Gopalakrishnan2011}. Coordination can be used to avoid pilot contamination \cite{Ashikhmin2012}, but only for a cluster of cells. Performance analysis in cellular systems as in \cite{Marzetta2010,Ashikhmin2012,Hoydis2013} considered a small number of base stations with a hexagonal topology. Recently, simple characterizations of cellular network performance were developed in \cite{Andrews2011b} using stochastic geometry where the base stations are randomly located. The work in \cite{Andrews2011b} does not readily extend to massive MIMO networks because  pilot contamination creates an additional self-interference that is not included in that original model.

In this letter, we analyze the performance of large-scale massive MIMO networks, where locations of the BSs are assumed to form a Poisson point process (PPP). First, we prove that
the signal-to-interference ratio (SIR) in a large-scale network still converges to an equivalent function, which only depends on the distance between interferers, as the number of antennas goes to infinity. Next,
we provide expressions for the asymptotic coverage probability and rate in both downlink and uplink by deriving the distribution of the equivalent function. Our analytical results show that the asymptotic rate achievable per user is superior to a comparable network with single antenna base stations, despite the presence of pilot contamination. The asymptotic sum rate per cell is also characterized and is found to be proportional to the length of the channel coherence time. Our results apply to massive MIMO systems using matched filtering beamforming and shared pilots. The extension to more sophisticated beamforming strategies in \cite{Hoydis2013} or to more coordination as in  \cite{Ashikhmin2012} are interesting topics of future work.

\section{System model}\label{system}
We consider a massive MIMO cellular network with perfect synchronization that operates under the following assumptions.
\begin{assumption}[TDD] \label{TDD} A three stage approach leveraging reciprocity due to TDD is assumed, as proposed in \cite{Marzetta2010}.
\begin{enumerate}
\item {\it Channel estimation}: Each user sends pre-assigned orthogonal pilot sequence of length $K$ to its BS. Denote the orthnormal pilot set as $\{\Phi_k\}_{1\le k\le K}$. The pilot set are assumed to be reused among all cells. 
\item {\it Uplink Data}: Users send uplink data to their BS. Each BS applies a maximum ratio combining receiver derived from the channel estimates.
\item {\it Downlink Data}: BSs send downlink data using matched-filter precoding.
\end{enumerate}
\end{assumption}
\begin{assumption}[PPP BS] The base stations (BSs) form a homogeneous PPP $\mathcal{N}_\mathrm{b}=\left\{X_{\ell}\right\}$ with density $\lambda_\mathrm{b}$ on the plane. Each base station is equipped with $M$ antennas. The network performance is examined when $M$ goes to infinity.
\end{assumption}
\begin{assumption}[PPP Users] During a block of time, each BS randomly schedules $K$ users in its cell. Each $k$-th user is assigned the $k$-th pilot sequence and the locations of all the users sharing the same pilot sequence are assumed to form an independent PPP with density $\lambda_\mathrm{b}$. Let  $Y^{(k)}_{i}$ denote the location of the $k$-th user in cell $\ell$, and $\Phi_\mathrm{u}^{(k)}=\left\{Y_{\ell}^{(k)}\right\}$ be the PPP of the users using the $k$-th pilot. The typical BS X0 is assumed to be associated with the nearest user in $\Phi_\mathrm{u}^{(k)}$. All users are equipped with one antenna.
\end{assumption}

Note that the actual locations of the users using the same pilot form a general stationary point process of density $\lambda_\mathrm{b}$, which is not Poisson, for the location of a user is not totally independent but restricted inside the Voronoi cell of its associated BS. The PPP user approximation, however, provides tractability in the uplink analysis.
\begin{assumption}[Independent Block Fading]
The channel is constant during one block and fades independently from block to block. Denote $\mathbf{h}^{(k)}_{\ell i}\in\mathbb{C}^{M\times1}$ as the channel vector from BS $X_\ell$ to $Y^{(k)}_{i}$. Each entry of $\mathbf{h}^{(k)}_{\ell i}$ satisfies $\mathbf{h}^{(k)}_{\ell i}(m)=\sqrt{\beta^{(k)}_{\ell i}}v^{(k)}_{\ell i}(m)$, where $\beta^{(k)}_{\ell i}$ is the path gain and shadowing from $X_\ell$ to $Y^{(k)}_{i}$, and $v^{(k)}_{\ell i}(m)$ accounts for the small-scale fading. In this letter, we assume independent Rayleigh fading between antennas, i.e. $\forall i,\ell,k$, and $m$, $v^{(k)}_{\ell i}(m)$ are independent complex random variables of distribution $\mathcal{CN}(0,1)$, and there is no shadowing. We use the log-distance path loss so that $\beta^{(k)}_{\ell i}=\left({r^{(k)}_{\ell i}}\right)^{-\alpha}$, where ${r^{(k)}_{\ell i}}$ is the distance from $X_\ell$ to $Y^{(k)}_{i}$, and $\alpha$ is the path loss exponent.
\end{assumption}
\begin{assumption}[No Noise]
The network is interference-limited, thus thermal noise is neglected in the analysis.
\end{assumption}
\begin{assumption}[No MMSE] The BSs estimate the channel by correlating the received training signal with the corresponding pilot and do not consider more sophisticated minimum mean squared error (MMSE) estimation as employed in  \cite{Hoydis2013}. Hence in the channel estimation stage, the received training signal by BS $X_\ell$ is $\mathbf{e}_\ell=\sum_{X_{\ell'}\in\mathcal{N}_\mathrm{b}}\sum_{1\le k\le K}\mathbf{h}_{\ell\ell'}^{(k)}\Phi_k^{*}.$
Its estimation of the channel to the $q$-th user $Y_\ell^{(k)}$ in its cell is
\begin{align*}
\mathbf{g}_{\ell\ell}^{(q)}=\mathbf{e}_\ell\Phi_q=\mathbf{h}_{\ell\ell}^{(q)}+\sum_{\ell'\ne\ell}\mathbf{h}_{\ell\ell'}^{(q)}.
\end{align*}
\end{assumption}
Note that $\mathbf{g}_{\ell\ell}^{(q)}$ is the estimation of $\mathbf{h}_{\ell\ell}^{(q)}$, and the estimation error is $\sum_{\ell'\ne\ell}\mathbf{h}_{\ell\ell'}^{(q)}$, which is caused by pilot contamination.

Assumptions 1-6 represent a simple massive MIMO systems in which each term of the SIR expression can be written in a simple form and analyzed using stochastic geometry. In the following sections, we examine the impact of the channel error due to pilot contamination in terms of coverage probability and rate in both downlink and uplink transmission.

\section{Asymptotic Analysis of Massive MIMO Network}
\subsection{Downlink Analysis}
To analyze the coverage probability in the downlink, a typical user denoted as $Y_{0}^{(1)}$, is located at the origin. Let $s_{\ell}^{(k)}$ be the data symbol that BS $X_\ell$ transmits to its $k$-th user $Y_{\ell}^{(k)}$. Assume that the data symbol $s_{\ell}^{(k)}$ are i.i.d. distributed with zero mean and unit variance. Let $\mathbf{w}_{\ell}^{(k)}$ denote the precoding vector at BS $X_\ell$ for the user $Y_{\ell}^{(k)}$. By Assumption \ref{TDD}, in the downlink, a BS uses the channel estimates obtained from the uplink training to do match-filter precoding. We first consider a simplified situation where there is no power constraint at the BS. In this case, let $\mathbf{w}_{\ell}^{(k)}=\mathbf{g}_{\ell}^{(k)}$. Then the typical user $Y_{0}^{(1)}$ at the origin receives
\begin{align*}
\hat{s}_{0}^{(1)}
&=\sum_{k=1}^{K}\sum_{X_\ell\in\mathcal{N}_\mathrm{b}}\mathbf{h}_{\ell0}^{(1)*}\mathbf{g}_{\ell}^{(k)}s_{\ell}^{(1)}\\
&=\sum_{k=1}^{K}\sum_{X_\ell\in\mathcal{N}_\mathrm{b}}\sum_{X_{\ell'}\in\mathcal{N}_\mathrm{b}}\mathbf{h}_{\ell0}^{(1)*}\mathbf{h}_{\ell\ell'}^{(k)}s_{\ell}^{(k)}.
\end{align*}
The downlink $\mbox{SIR}_\mathrm{DL}$ can be expressed as
\begin{align*}\label{downsir}
\frac{\|h_{00}^{(1)}\|^4}{\sum_{\ell\ne0}\|h_{\ell0}^{(1)}\|^4+\sum_{\ell'\ne0,\ell}|h_{\ell0}^{(1)*}h_{\ell\ell'}^{(1)}|^2+\sum_{k\ne1,\ell,\ell'}|h_{\ell0}^{(1)*}h_{\ell\ell'}^{(k)}|^2}.
\end{align*}
When the number of BS antennas goes to infinity, we have the following convergence theorem.
\begin{theorem}\label{thm:downsir}
When $M\to\infty$, $\mathrm{SIR}_\mathrm{DL}\stackrel{a.s.}{\to}\frac{\left(\beta_{00}^{(1)}\right)^2}{\sum_{\ell\ne0}\left(\beta_{\ell0}^{(1)}\right)^2}$.
\end{theorem}
\begin{proof}
See Section \ref{proof1}.
\end{proof}
We define the coverage probability of the downlink $\bbP_\mathrm{DL}(T)$ as the probability that the received SIR by the typical user large than a threshold $T$. Assuming that $M$ is large enough, the asymptotic downlink coverage probability can be obtained by deriving the distribution of $\left(\beta_{00}^{(1)}\right)^2/\sum_{\ell\ne0}\left(\beta_{\ell0}^{(1)}\right)^2$.
\begin{theorem}\label{thm:distrib}
When $M\to\infty$, the downlink coverage probability is
\begin{align}
P_\mathrm{DL}(T)=\int_{\mathbb{R}}\frac{\mathrm{e}^{\frac{2j\pi s}{T}}-1}{2j\pi s\;\eta(2j\pi s)}\mathrm{d}s,
\end{align}
where $\eta(x)=\mathrm{e}^{-x}+x^{\frac{1}{\alpha}}\gamma\left(1-\frac{1}{\alpha},x\right)$, and $\gamma(a,z)=\int_{0}^{z}t^{a-1}\mathrm{e}^{-t}\mathrm{d}t$ is the lower incomplete gamma function.
\end{theorem}
\begin{proof}
See Section \ref{proof2}.
\end{proof}
Note that the coverage probability is invariant with the BS density. When the threshold $T>1$, we can compute $P_\mathrm{DL}(T)$ in closed form.
\begin{corollary}\label{cor:closed}
When $T>1$, $P_\mathrm{DL}(T)=\frac{\alpha\sin(\pi/\alpha)}{\pi T^{1/\alpha}}.$
\end{corollary}
\begin{proof}
See Remark 9 in \cite{Blaszczyszyn2013}, and substitute the equivalent path loss exponent as $2\alpha$ in the massive MIMO case.
\end{proof}
Note that in Theorem \ref{thm:downsir}, the equivalent path loss exponent doubles in the asymptotic SIR expression. The expression in Corollary \ref{cor:closed} also works for single-antenna networks, and is an increasing function of $\alpha\ge2$. Hence, massive MIMO networks provide better coverage probability asymptotically than single-antenna systems when $T\ge1$.

Next, we investigate the downlink coverage probability when a fixed transmission power constraint is assumed at each BS. In this case, we have to normalize the precoding vector as $\overline{\mathbf{w}}_{\ell}^{(k)}={\mathbf{g}_{\ell}^{(k)}} / {\left\|\mathbf{g}_{\ell}^{(k)}\right\|}$. 
\begin{theorem}
Let $\overline{\mathrm{SIR}}_\mathrm{DL}$ denote the downlink SIR when considering the power constraints at BSs, then it follows that
\begin{align}\label{pcsir}
\lim_{M\to\infty}\overline{\mathrm{SIR}}_\mathrm{DL}\stackrel{a.s.}{\to}\frac{\left(\beta_{00}^{(1)}\right)^2/b_0^{(1)}}{\sum_{\ell\ne0}\left(\beta_{\ell0}^{(1)}\right)^2/b_\ell^{(1)}},
\end{align}
where $b_\ell^{(k)}=\sum_{\ell'\ge0}\beta_{\ell\ell'}^{(k)}$.
\end{theorem}
The proof is similar to that of Theorem \ref{thm:downsir}. Note that the denominator $b_\ell^{(k)}$ in the sum introduces high correlations between terms, which renders the distribution of (\ref{pcsir}) difficult to compute. Simulations in Section \ref{simu} show that the downlink coverage probability evaluated without power constraints is generally an upper bound for the power-constrained cases.

Based on the results of coverage probability, we can compute the asymptotic downlink achievable rate.
\begin{theorem}\label{rate}
When $M$ goes to infinity, the achievable rate $\Gamma_\mathrm{DL}$ can be computed as
\begin{align}
\Gamma_\mathrm{DL}=\int_{T>0}\frac{P_\mathrm{DL}(T)}{1+T}\mathrm{d}T.
\end{align}
\end{theorem}
As shown in Section \ref{simu}, the asymptotic downlink rate per user in the massive MIMO networks is larger than the rate of the baseline network where each BS is equipped with 1 antenna, and serves 1 user at a time.
\subsection{Uplink Analysis}
Now we denote $s_{\ell}^{(k)}$  as the data symbol that the user $Y_{\ell}^{(k)}$ transmits to the BS $X_\ell$. In the uplink transmission stage, we fix a typical BS $X_0$ at the origin. Then $X_0$ receives
\begin{align*}
\mathbf{u}_{0}=\sum_{1\le k\le K}\sum_{X_{\ell}\in\mathcal{N}_\mathrm{b}}\mathbf{h}_{0\ell}^{(k)}s_{\ell}^{(k)}.
\end{align*}
To decode the uplink data from a mobile user, e.g user $Y_0^{(1)}$, BS $X_0$ uses the channel estimates obtained from channel estimation stage to do maximum ratio combining as
\begin{align*}
\hat{s}_{0}^{(1)}&=\left(\mathbf{g}_{00}^{(1)}\right)^{*}\mathbf{u}_{0}\\
&=\sum_{k=1}^{K}\sum_{X_\ell,X_{\ell'}\in\mathcal{N}_\mathrm{b}}\mathbf{h}_{0\ell}^{(1)*}\mathbf{h}_{0 \ell'}^{(k)}s_{\ell}^{(k)}.
\end{align*}
The corresponding uplink $\mathrm{SIR}_\mathrm{UL}$ for $\hat{s}_{0}^{(1)}$ is
\begin{align*}
\frac{\|h_{00}^{(1)}\|^4}{\sum_{\ell\ne0}\|h_{0\ell}^{(1)}\|^4+\sum_{\ell\ne\ell'}|h_{0\ell}^{(1)'}h_{0\ell'}^{(1)}|^2+\sum_{k\ne1}\sum_{\ell',\ell}|h_{0\ell}^{(1)'}h_{0\ell'}^{(k)}|^2}
\end{align*}
Similar to the downlink SIR, when $M\to\infty$ we have the following theorem.
\begin{theorem}\label{thm:sir}
When $M\to\infty$, $\mathrm{SIR}_\mathrm{UL}\stackrel{a.s.}{\to}\frac{\left(\beta_{00}^{(1)}\right)^2}{\sum_{\ell\ne0}\left(\beta_{0\ell}^{(1)}\right)^2}$.
\end{theorem}
Under the PPP user assumption, the uplink coverage probability and rate are derived in the following corollary.
\begin{corollary}
Under the PPP user approximation, the asymptotic uplink coverage probability $P_\mathrm{UL}(T)=P_\mathrm{DL}(T)$, and the asymptotic uplink rate per user $\Gamma_\mathrm{UL}=\Gamma_\mathrm{DL}$, where $P_\mathrm{DL}(T)$ and $\Gamma_\mathrm{DL}$ are derived in Theorem \ref{thm:distrib} and Theorem \ref{rate} respectively.
\end{corollary}
\begin{proof}
Under the PPP user approximation
\begin{align*}
P_\mathrm{UL}(T)&=
\mathbb{P}_{Y^{(1)}_0}\left[\frac{\left(\beta_{00}^{(1)}\right)^2}{\sum_{\ell\ne0}\left(\beta_{\ell0}^{(1)}\right)^2}>T\right]\\
&=\mathbb{P}_{X_0}\left[\frac{\left(\beta_{00}^{(1)}\right)^2}
{\sum_{\ell\ne0}\left(\beta_{0\ell}^{(1)}\right)^2}>T\right]\\
&=P_\mathrm{DL}(T),
\end{align*}
where $\bbP_{X}$ is the conditional probability given that $X$ is located at the origin.
\end{proof}

\subsection{Sum Rate}

Now we examine the sum rate in a cell. Given there are $L$ channel uses in a time block, the asymptotic sum throughput of a cell in a time block is $\Gamma_{\mathrm{tot}}=K(L-K)\Gamma_\mathrm{DL}/L,$ where $K$ is the length of the pilot sequence and equivalently the number of users served simultaneously by a BS. Note that due to the large number of BS antennas, all $K$ users in a cell can be served free of intra-cell interference simultaneously. The rate per user $\Gamma_{\mathrm{DL}}$ in massive MIMO networks is independent of $K$. Hence, to achieve the optimal asymptotic throughput, we let $K=\lfloor \frac{L}{2}\rfloor $, and the optimal average sum rate in a time block is $\Gamma_{\mathrm{tot}}^{*}\approx\frac{L}{4}\Gamma_\mathrm{DL},$ which scales with the coherence time.

\section{Simulations}\label{simu}
We simulate the downlink coverage probability in Fig. \ref{fig:coverage} and compare with the analytical form. We choose the baseline model as the network consisting of single-antenna BSs. When $\alpha>2$, simulations illustrate that the coverage probability in the massive MIMO network is better than that in the baseline network. 
It is shown that when power constraints is considered in the downlink, the asymptotic coverage probability is worse than the ideal situation with no power constraint.

Next, we compare the asymptotic rate per user and the average sum rate of a cell between the massive MIMO network and the baseline network in Table \ref{table:rate}. It can be concluded that though limited by pilot contamination, the asymptotic rate per user with many antennas is still much larger than that in the network with single-antenna BSs. Moreover, the gain of the average sum rate is even larger in massive MIMO networks, for with a large number of antennas, a BS can serve multiple users simultaneously without intra-cell interference.
\begin{table}
\begin{center}
\caption{Comparison of Achievable Rate With $\alpha=4$ and $L=16$}\label{table:rate}
\begin{tabular}{ c | c | c  }
\hline
    \mbox{Number of antennas at BSs} & \mbox{Infinity} & 1 \\ \hline
    \mbox{Achievable rate per User(bps/Hz)} &3.79 & 2.15\\ \hline
  \mbox{Average sum rate per cell (bps/Hz)} & 15.16& 2.15\\ \hline
    \end{tabular}
  \end{center}
\end{table}
\begin{figure} [!ht]
\centerline{
\includegraphics[width=0.8\columnwidth]{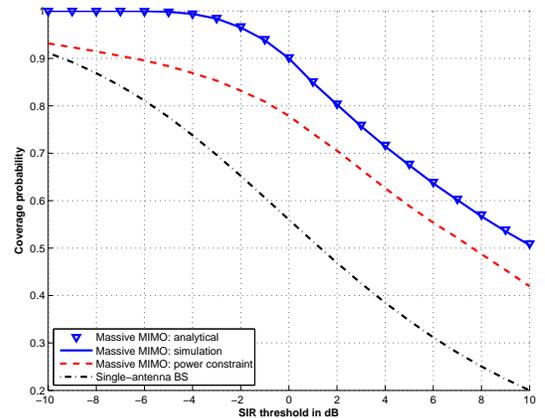}
   }
\caption{Comparison of coverage when $\alpha=4$. The asymptotic coverage probability with power constraint is plotted by Monte Carlos simulations.}\label{fig:coverage}
\end{figure}
\section{Conclusions}
In this letter, we proposed analytical expressions to evaluate the asymptotic coverage probability in both downlink and uplink in massive MIMO networks. When the SIR threshold is larger than 1, the asymptotic coverage probability can be expressed in closed form. The asymptotic rate per user is computed, and shown better than the network with single antenna BS. Moreover, we also showed that the total sum rate in a cell is in proportion to the length of the coherent time. More practical issues, such as space correlation between antennas and power control in the downlink, are expected to  be incorporated in the future analysis.
\section{Appendix}
\subsection{Proof Sketch of Theorem \ref{thm:downsir}}\label{proof1}
To prove Theorem \ref{thm:downsir}, it is equivalent to show the following:
\begin{align}
\lim_{M\to\infty}\frac{\|h_{00}^{(1)}\|^4}{M^2}-\left(\beta_{00}^{(1)}\right)^2\stackrel{a.s.}{=}0,\label{eqn:1}\\
\lim_{M\to\infty}\frac{\sum_{\ell\ne0}\|h_{\ell0}^{(1)}\|^4}{M^2}-\sum_{\ell\ne0}\left(\beta_{\ell0}^{(1)}\right)^2\stackrel{a.s.}{=}0,\label{eqn:2}\\
\lim_{M\to\infty}\frac{\sum_{\ell'\ne0,\ell}|h_{\ell0}^{(1)*}h_{\ell\ell'}^{(1)}|^2}{M^2}\stackrel{a.s.}{=}0,\label{eqn:3}\\
\lim_{M\to\infty}\frac{\sum_{k\ne1}\sum_{\ell,\ell'}|h_{\ell0}^{(1)*}h_{\ell\ell'}^{(k)}|^2}{M^2}\stackrel{a.s.}{=}0.\label{eqn:4}
\end{align}
\begin{proof}
The proof of (\ref{eqn:1}) is a direct application of the strong law of large numbers (SLLN), e.g see \cite{Ashikhmin2012}.

To prove (\ref{eqn:2}), we need to show that $\forall \epsilon,\delta>0$, $\exists M_0>0$, such that $\forall m>M_0$, we have
\begin{align*}
\bbP\left[\left|\frac{\sum_{\ell\ne0}\|h_{\ell0}^{(1)}\|^4}{m^2}-\sum_{\ell\ne0}\left(\beta_{\ell0}^{(1)}\right)^2\right|<\epsilon\right]>1-\delta.
\end{align*}
Let $\delta_0=1-(1-\delta)^{1/3}$. First we show that there exists a ball $\mathcal{B}(0,R_0)$ large enough, such that
\begin{align}\bbP\left[\sum_{X_\ell\notin\mathcal{B}(0,R_0)}\left|\frac{\|h_{\ell0}^{(1)}\|^4}{M^2}-\left(\beta_{\ell0}^{(1)}\right)^2\right|>\epsilon/3\right]<\delta_0.\label{eqn:p1}\end{align}
 By Campbell's formula, it holds that for all $M\ge1$
\begin{align*}
\bbE\left[\sum_{X_\ell\notin\mathcal{B}(0,R_0)}\left|\frac{\|h_{\ell0}^{(1)}\|^4}{M^2}-\left(\beta_{\ell0}^{(1)}\right)^2\right|\right]<\frac{3\pi\lambda_\mathrm{b}}{2\alpha-2}R_0^{2-2\alpha}.
\end{align*}
By Markov's inequality, we can always find $R_0$ large enough such that (\ref{eqn:p1}) is satisfied.

Next, it can be shown that $\forall R_0>0$, there exists $N$ large enough such that there are at most $N$ BSs inside $\mathcal{B}(0,R_0)$ with probability $1-\delta_0$. Then given there are at most $N$ BSs inside $\mathcal{B}(0,R_0)$, by the SLLN, there exists $M_0$ large enough such that $\forall m>M_0$, it holds that \begin{align*}\bbP\left[\sum_{X_\ell\in\mathcal{B}(0,R_0)}\left|\frac{\|h_{\ell0}^{(1)}\|^4}{m^2}-\left(\beta_{\ell0}^{(1)}\right)^2\right|>\epsilon/3\right]<\delta_0.\end{align*}
Finally, it follows that
\begin{align*}
\bbP\left[\left|\frac{\sum_{\ell\ne0}\|h_{\ell0}^{(1)}\|^4}{m^2}-\sum_{\ell\ne0}\left(\beta_{\ell0}^{(1)}\right)^2\right|<\epsilon\right]>(1-\delta_0)^3=1-\delta.
\end{align*}

To prove (\ref{eqn:3}), we assume that no users and BSs are co-located, i.e $\forall r_{\ell\ell'}^{(k)}>\Delta$ for some constant $\Delta>0$. First, we can show that there exists a ball $\mathcal{B}(0,R_0)$ large enough, such that \begin{align}\label{eqn:proof2}\bbP\left[\sum_{X_\ell\notin\mathcal{B}(0,R_0),X_\ell} \frac{|h_{\ell0}^{(1)*}h_{\ell\ell'}^{(1)}|^2}{M^2}\ge\frac{\epsilon}{3}\right]<\delta/3.
\end{align}
It follows that
\begin{align*}
&\bbE\left[\sum_{X_\ell\notin\mathcal{B}(0,R_0)} \frac{|h_{\ell0}^{(1)*}h_{\ell\ell'}^{(1)}|^2}{M^2}\right]\\
&\stackrel{(a)}{=}\frac{2\lambda_\mathrm{b}^2}{M^2}\int_{\overline{\mathcal{B}(0,R_0)}}\int_{\bbR^2}\frac{1}{|x|^{\alpha}}\frac{1}{|x-y|^{\alpha}}\mathrm{d}x\mathrm{d}y\\
&{\le}\frac{8\pi^2\lambda_\mathrm{b}^2}{M^2}\int_{R_0}^{\infty}\frac{1}{r^{\alpha-1}}\mathrm{d}r\int_{\Delta}^{\infty}\frac{1}{t^{\alpha-1}}\mathrm{d}t\\
&\le\frac{8\pi^2\lambda_\mathrm{b}^2}{M^2(\alpha-2)^2\Delta^{\alpha-2}R_0^{\alpha-2}},
\end{align*}
where step $(a)$ is from Campbell's formula for factorial moment of PPP. By Markov's inequality, we can always choose $R_0$ large enough such that (\ref{eqn:proof2}) is satisfied.

Next, we can show that there exists $N$ large enough, such that there are at most $N$ BSs in $\mathcal{B}(0,R_0)$ with probability $1-\delta/3$. Given there are at most $N$ BSs in  $\mathcal{B}(0,R_0)$, we can show that for $M$ large enough,
\begin{align*}\bbP\left[\sum_{X_\ell\in\mathcal{B}(0,R_0),X_\ell} \frac{|h_{\ell0}^{(1)*}h_{\ell\ell'}^{(1)}|^2}{M^2}\ge\frac{\epsilon}{3}\right]<\delta/3.
\end{align*}
Hence, using union bound, we can prove that
\begin{align*}\bbP\left[\sum_{\ell'\ne0,\ell} \frac{|h_{\ell0}^{(1)*}h_{\ell\ell'}^{(1)}|^2}{M^2}\ge\epsilon\right]<\frac{\delta}{3}+\frac{\delta}{3}+\frac{\delta}{3}=\delta.
\end{align*}
Similarly, we can prove the convergence of (\ref{eqn:4}) following the same steps as the proof of (\ref{eqn:3}).
\end{proof}

\subsection{Proof of Theorem \ref{thm:distrib}}\label{proof2}
Define $f=1/\mbox{SIR}_\mathrm{DL}$ as the inverse of the downlink SIR. By \cite{Blaszczyszyn2013}, the Laplace transform of $f$ is $\bbE\left[\mathrm{e}^{-zf}\right]=1/\eta(z)$. By Parseval's theorem, it follows
\begin{align*}
P_\mathrm{DL}(T)=\bbP(f<1/T)=\int_{\mathbb{R}}\frac{\mathrm{e}^{\frac{2j\pi s}{T}}-1}{2j\pi s\;\eta(2j\pi s)}\mathrm{d}s.
\end{align*}
\bibliographystyle{IEEEtran}

\end{document}